\documentclass[12pt]{article} \usepackage{fullpage}
\usepackage{amssymb}
\usepackage{amsmath}
\usepackage{amsthm}
\usepackage{todonotes}

\usepackage{nicefrac}
\usepackage{url}
\usepackage{color}
\usepackage[american]{babel}
\usepackage{graphicx}
\usepackage{subcaption}
\newcommand{\student}[1]{} 

\newtheorem{theorem}{Theorem}

\newtheorem{definition}{Definition}

\newtheorem{observation}{Observation}

\newtheorem{lemma}{Lemma}

\newcommand{\leaveout}[1]{}

\renewcommand{\medskip}{\smallskip}
\renewcommand{\int}{{\ensuremath{\rm int\,}}}

\date{}
\title{A 2-Approximation for the Height of Maximal Outerplanar Graph Drawings }
\author{Therese Biedl\thanks{David R.~Cheriton School of Computer
Science, University of Waterloo, Waterloo, Ontario N2L 1A2, Canada.
{\it biedl@uwaterloo.ca}.
Supported by NSERC.}
\and Philippe Demontigny\thanks{{\it phdemontigny@gmail.com}.  
Part of this work appeared as the author's
Master's thesis at UWaterloo.}}

\begin{document}

\maketitle
\begin{abstract}
In this paper, we study planar drawings of maximal outerplanar graphs
with the objective of achieving small height.
A recent paper gave an algorithm for such drawings
that is within a factor of 4 of the optimum height.
In this paper, we substantially improve the approximation factor to become 2.
The main ingredient is to define a new parameter of outerplanar graphs
(the so-called {\em umbrella depth}, obtained by recursively splitting
the graph into graphs called umbrellas).  We argue that the height of
any poly-line drawing must be at least the umbrella depth, and then devise
an algorithm that achieves height at most twice the umbrella depth.
\end{abstract}

\section{Introduction}

Graph drawing is the art of creating a picture of a graph that is visually appealing.   
%
%
In this paper, we are interested in drawings of so-called {\em outer-planar
graphs}, i.e., graphs that can be drawn in the plane such that no two edges
have a point in common (except at common endpoints) and all vertices are 
incident to the outer-face.   All drawings are required to be planar, i.e.,
to have no crossing.  The drawing model used is that of 
flat visibility representations where vertices are horizontal segments and
edges are horizontal or vertical segments, but any such drawing can be
transformed into a poly-line drawing (or even a straight-line drawings if the
width is of no concern) without adding height \cite{Bie-GD14}.

Every planar graph has a straight-line drawing in an $n\times n$-grid \cite{Sch90,FPP90}.
Minimizing the area is NP-complete \cite{KW07}, even for
outer-planar graphs \cite{Bie14-ICALP}. 
In this paper, we focus on minimizing just one direction of a drawing
(we use the height; minimizing the width is equivalent after rotation). 
It is not known whether minimizing the height of a planar drawing is
NP-hard (the closest related result concerns minimizing the height if
edges must connect adjacent rows \cite{HR92}).
Given the height $H$, testing whether a planar drawing of height $H$ exists 
is fixed parameter tractable in $H$ \cite{DFK+08}, but
the run-time is exceeding large in $H$.  As such, approximation algorithms
for the height of planar drawings are of interest.

It is known that any graph $G$ with a planar drawing of height $H$ has 
$pw(G)\leq H$ \cite{FLW03}, where $pw(G)$ is the so-called pathwidth of $G$. 
This makes the pathwidth 
a useful parameter for approximating the height of a planar 
graph drawing.  For a tree $T$, Suderman gave an algorithm to draw $T$ with
height at most $\lceil \frac{3}{2}pw(T)\rceil$ \cite{Sud04},
making this an asymptotic $\frac{3}{2}$-approximation algorithm.
It was discovered later that optimum-height drawings can be found efficiently
for trees \cite{MAR11}.  Approximation-algorithms for the height or width of
order-preserving and/or upward tree drawing have also been investigated
\cite{ASR+10,BB16-ArXiV,Bie-OPTI-ArXiV}.

For outer-planar graphs, the first author gave two results that will be improved
upon in this paper.   In particular, every maximal outerplanar graph has a
drawing of height at most $3\log n-1$ \cite{Bie-GD02} and of height $4pw(G)-3$
\cite{Bie-WAOA12}.
Note that the second result gives a 4-approximation on the height of drawing
outerplanar graphs, and improving this ``4'' is the main objective of the
current paper.
A number of results for drawing outer-planar graphs have been developed since paper \cite{Bie-GD02}.
In particular, 
any outerplanar graph with maximum degree $\Delta$ admits a planar straight-line drawing with area $O(\Delta n^{1.48})$ \cite{GR07},
or with area $O(\Delta n\log n)$ \cite{Frati12}.
The former bound was improved to $O(n^{1.48})$ area \cite{DBF09}. 
Also, every so-called balanced outer-planar graph can be drawn in an
$O(\sqrt{n})\times O(\sqrt{n})$-grid \cite{DBF09}.


In this paper, we present a 2-approximation algorithm for the height of
planar drawings of maximal outer-planar graphs.  The key ingredient is to
define the so-called {\em umbrella depth} $ud(G)$ in
Section~\ref{chap:Umbrellas}. In Section~\ref{chap:UpperBound}, we show 
that any outerplanar graph $G$ has a planar drawing of height at most $2ud(G)+1$. 
(We actually show a height of $2bd(G)+1$, where the {\em bonnet depth} $bd(G)\leq ud(G)$
is another newly defined graph parameter.)  This algorithm is a relatively minor
modification of the one in \cite{Bie-WAOA12}, albeit described differently.
The bulk of the work for proving a better approximation factor hence lies
in proving a better lower bound, which we do in Section~\ref{chap:LowerBound}:
Any maximal outerplanar graph $G$ with a planar drawing of height $H$ has $ud(G)\leq H-1$.
This proves that our result is a 2-approximation for the optimal height, which 
must fall in the range $[ud(G)+1, 2ud(G)+1]$.  


\section{Preliminaries}\label{chap:Preliminaries}

Throughout this paper, we assume that $G=(V,E)$ is a simple graph with $n\geq 3$
vertices and $m$ edges that is {\em maximal outer-planar}.  Thus, $G$ has
a \emph{standard planar embedding} in which all vertices are in the 
\emph{outer face} (the infinite connected region outside the drawing)
and form an $n$-cycle, and all \emph{interior faces} are triangles.
We call an edge $(u,v)$ of $G$ a \emph{cutting edge} if $G-\{u,v\}$ is
disconnected, and a \emph{non-cutting edge} otherwise.%
\footnote{The cutting edges are
exactly those edges for which in the standard embedding both incident
faces are interior, but we prefer to phrase this and the following definitions
independent of the standard embedding since we do not necessarily draw the
graph in the standard embedding.}
In an outer-planar graph, any cutting edge $(u,v)$ 
has exactly two \emph{cut-components}, i.e., there are two maximal
outerplanar subgraphs 
$G_1,G_2$ of $G$ such that $G_1\cap G_2=\{u,v\}$ and $G_1\cup G_2=G$.

\begin{figure}[h!]
\hspace*{\fill}
\begin{subfigure}[b]{0.3\textwidth}
	\includegraphics[width=\linewidth,trim=0 0 270 0,clip]{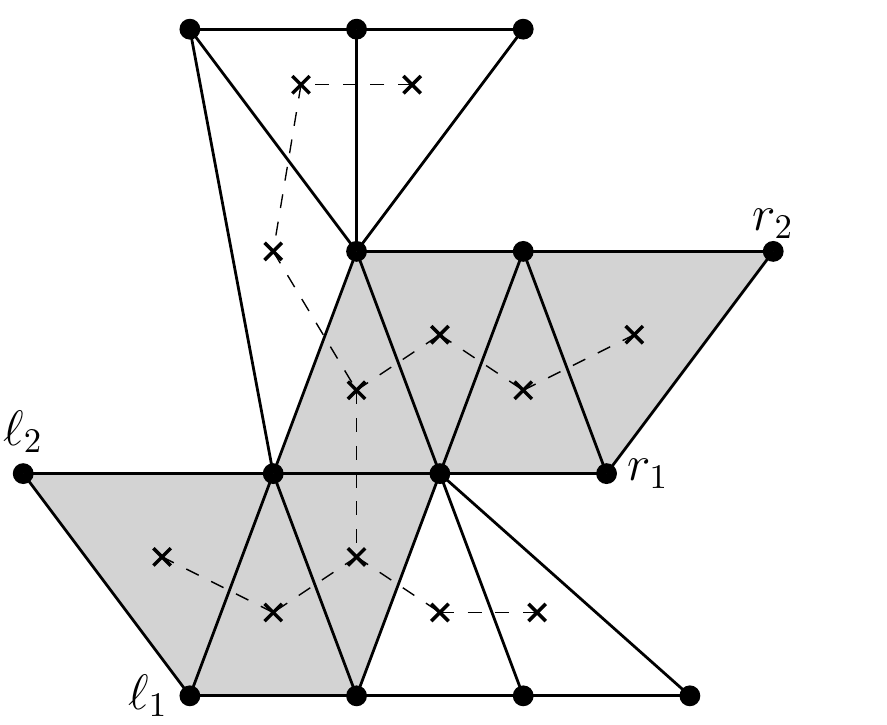}
	\caption{}
\end{subfigure}
\hspace*{\fill}
\begin{subfigure}[b]{0.3\linewidth}
	\includegraphics[width=\linewidth,trim=250 0 0 0,clip]{DefinitionsExampleFigure.pdf}
	\caption{}
\end{subfigure}
\hspace*{\fill}
\caption{(a) A straight-line drawing in the standard embedding, including the dual tree (dashed edges) and an outerplanar path (shaded) connecting $(\ell_1,\ell_2)$ with $(r_1,r_2)$. (b) A flat visibility representation. Both drawings have height 4. \label{fig:DefinitionsExample}}
\end{figure}

The \emph{dual tree} $T$ of $G$ is the weak dual graph of $G$ in the standard embedding, i.e., $T$ has a vertex for each interior face of $G$, and an edge
between two vertices if their corresponding faces in $G$ share an edge. 
An \emph{outerplanar path} $P$ is a maximal outerplanar graph whose dual tree is a path.   
We say that $P$ {\em connects edges $e$ and $e'$}
if $e$ is incident to $f_1$ and 
$e'$ is incident to $f_k$, where $f_1$ and $f_k$ are the first and last face 
in the path that is the dual tree of $P$.
An outerplanar path $P$ with $n=3$ is a triangle and 
connects any pair of its edges.
Since any two interior faces are connected by a path in the dual tree, there exists
an outerplanar path connecting $e$ and $e'$ for any two edges $e,e'$.

\medskip\noindent{\bf Graph drawing: }
A \emph{drawing} of a graph consists of a point or an axis-aligned box for every vertex, and a polygonal curve for every edge. We only consider \emph{planar drawings} where none of the points, boxes, or curves intersect unless the corresponding elements do in the original graph. In this paper, a planar drawing is not required to reflect a graph's given planar embedding. In a \emph{flat visibility representations} vertices are represented by horizontal line segments, and edges are vertical or horizontal straight-line segments. (For ease of reading, draw vertices as boxes of small height in our illustrations.) In a \emph{poly-line drawing} vertices are 
points and edges are polygonal curves, while in a {\em straight-line drawing}
vertices are points and edges are line segments. 
In this paper, we only study planar flat visibility representations, but simply speak of a {\em planar drawing},
because it is known that any planar flat visibility representation can be converted into a planar
straight-line drawing of the same height and vice versa \cite{Bie-GD14}.

We require that all defining features (points, endpoints of segments, bends)
are placed at points with integer $y$-coordinates.  A {\em layer} (or {\em row})
is a horizontal line with integer $y$-coordinate that intersects elements of the
drawing, and the {\em height} is the number of layers.
We do not enforce integer
$x$-coordinates since we do not focus on minimizing the width. We
can always achieve $O(n)$ width (without adding height) for 
visibility representations 
and for the poly-line drawings obtained from them \cite{Bie-GD14}.)

\section{Umbrellas, bonnets and systems thereof}
\label{chap:Umbrellas}

In this section, we introduce 
a method of splitting maximal outerplanar graphs into systems of special
outerplanar graphs called \emph{umbrellas} and \emph{bonnets}.

\begin{definition} \label{def:Umbrella}
Let $G$ be a maximal outer-planar graph, 
let $U$ be a subgraph of $G$ with $n \geq 3$, and let $(u,v)$
be a non-cutting edge of $G$.  We say that $U$ is an
\emph{umbrella with cap $(u,v)$} if it can be written as the union
of three outerplanar paths $P$, $F_1$, and $F_2$ such that:
\begin{enumerate}
\item $P$ (the {\em handle}) connects $(u,v)$ to some other non-cutting edge of $G$,
\item $F_1$ (the {\em fan at $u$}) contains only $u$ and neighbours of $u$.
$F_2$ (the {\em fan at $v$}) contains only $v$ and neighbours of $v$.
\item $F_1$ and $F_2$ are edge-disjoint.  $F_1$ and $P$ have exactly one edge (incident to $u$) in common; $F_2$ and $P$ have exactly one edge (incident to $v$) in common.
\item All neighbours of $u$ and $v$ belong to $U$.
\end{enumerate}
\end{definition}

See also Figure~\ref{fig:UmbrellaSystemExample}(a).
We allow the fans to be empty,
but $P$
must have at least one interior face (the one incident to $(u,v)$).
Any edge $(a,b)$ of $U$ that is a cutting edge
of $G$, but not of $U$, is called an {\em anchor-edge} of $U$ in $G$.  (In the
standard embedding, such edges are on the outer-face of $U$ but not on the
outer-face of $G$.)  
The {\em hanging subgraph with respect to anchor-edge $(a,b)$ of $U$ in $G$} 
is the cut-component $S_{a,b}$ of $G$ with respect to cutting-edge $(a,b)$ that does not 
contain the cap $(u,v)$ of $U$.  We often omit ``of $U$ in $G$'' when umbrella and
super-graph are clear from the context.

%

\begin{definition} \label{def:UmbrellaSystem}
Let $G$ be a maximal outerplanar graph with $n\geq 3$, and let $(u,v)$ 
be a non-cutting edge of $G$.
An \emph{umbrella system $\mathcal{U}$ on $G$ with root-edge
$(u,v)$} is a collection 
$\mathcal{U}=\{U_0\} \cup \mathcal{U}_1 \cup \dots \cup \mathcal{U}_k$
of subgraphs of $G$ for some $k\geq 0$ that satisfy the following:
\begin{enumerate}
\item $U_0$ (the \emph{root umbrella}) is an umbrella with cap $(u,v)$,
\item $U_0$ has $k$ anchor-edges. We denote them by $(u_i,v_i)$ for $i=1,\dots,k$, and let $S_i$ be the hanging subgraph with respect to $(u_i,v_i)$.
\item For $i=1,\dots,k$, $\mathcal{U}_i$ (the {\em hanging umbrella system}) 
	is an umbrella system 
	of $S_i$ with root-edge $(u_i,v_i)$.
\end{enumerate}
The {\em depth} of such an umbrella system is $d(\mathcal{U}):=
1+\max_{i} d(\mathcal{U}_i)$.  
\end{definition}

\begin{figure}[ht!]
\centering
\hspace*{\fill}
\begin{subfigure}[b]{0.4\linewidth}
	\includegraphics[width=\linewidth,page=1]{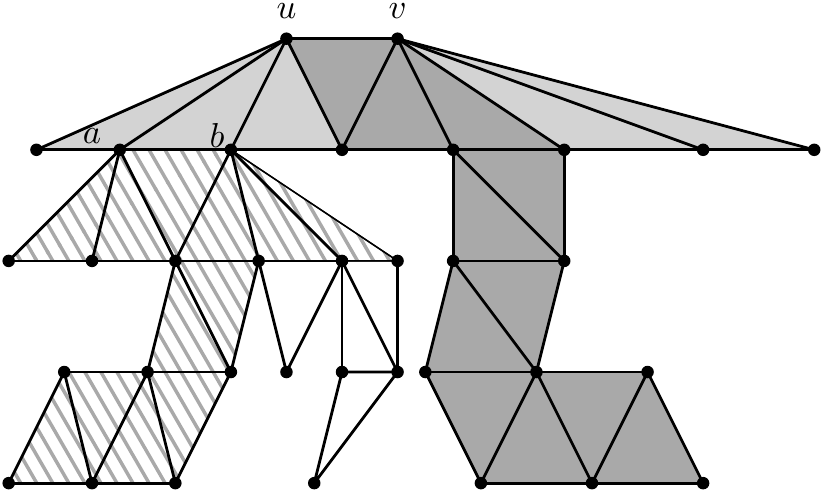}
	\caption{}
\end{subfigure}
\hspace*{\fill}
\begin{subfigure}[b]{0.4\linewidth}
	\includegraphics[width=\linewidth,page=2]{UmbrellaSystemExampleFigure.pdf}
	\caption{}
\end{subfigure}
\hspace*{\fill}
\caption{(a) An umbrella system of depth 3.  The root umbrella is shaded, with its handle darker shaded.  (b) The same graph has a bonnet system of depth 2, with the root bonnet shaded and its ribbon darker shaded.   }
\label{fig:UmbrellaSystemExample}
\label{fig:BonnetSystemExample}
\end{figure}

Define $ud(G;u,v)$
(the {\em (rooted) umbrella depth} of $G$) to be the maximum depth over
all umbrella systems with root-edge $(u,v)$.  
Note that the umbrella depth
depends on the choice of the root-edge; define the {\em free
umbrella depth} $ud^{\mathit free}(G)$ to be the minimum umbrella depth 
over all possible root-edges.
(One can show that the free umbrella depth is at most one unit less than
the rooted umbrella depth for any choice of root-edge; see the appendix.)

\medskip\noindent{\bf Bonnets: } A {\em bonnet} is a generalization of an umbrella
that allows two handles, as long as they go to different sides of the face $(u,v)$.
Thus, condition (1) of the definition of an umbrella gets replaced by 
\begin{quotation}
(1') $P$ (the {\em ribbon}) 
connects two non-cutting edges and contains $u$, $v$ and their
common neighbour.
\end{quotation}
Other than that, bonnets are defined exactly like umbrellas.  See also
Figure~\ref{fig:BonnetSystemExample}(b).  We can also define a {\em bonnet system},
{\em root bonnet}, etc., exactly as for
an umbrella system, except that ``bonnet'' is
substituted for ``umbrella'' everywhere.  
Let $bd(G;u,v)$ (the {\em rooted bonnet-depth
of $G$}) be the minimum possible depth of
a bonnet system with root-edge $(u,v)$, and let $bd^{\mathit free}(G)$ be
the minimum bonnet-depth over all choices of root-edge.  Since any umbrella
is a bonnet, clearly $bd(G;u,v)\leq ud(G;u,v)$ for all root-edges $(u,v)$.

We would like to emphasize that the root bonnet $U_0$ of a 
bonnet system must contain {\em all}
edges incident to the ends $u,v$ of the root-edge.  If follows that no
edge incident to $u$ or $v$ can be an anchor-edge of $U_0$, else the
hanging subgraph at it would contain further neighbours of $u$ (resp.~$v$).
We note this trivial but useful fact for future reference:
\begin{observation} 
\label{obs:HangingSubgraphLocations}
In a bonnet system with root-edge $(u,v)$, no edge incident to $u$ or $v$ is an anchor-edge of the root bonnet.
\end{observation}

\section{From Bonnet System to Drawing}\label{chap:UpperBound}

In this chapter, we show how to create a flat visibility representation, given a maximal outerplanar graph $G$ and a bonnet system of $G$. The drawings we create will not be in the standard embedding of $G$,
as we will place drawings of hanging subgraphs
inside an inner face of the root bonnet. 
For merging purposes, we draw the root-edge $(u,v)$ in a special way: It
\emph{spans} the top layer, which means that $u$ touches the top left corner of the drawing, and $v$ touches the top right corner, or vice versa (see for example Figure \ref{fig:HandleStandardEmbedding}(b)).
We first explain how to draw the root bonnet.

\begin{lemma} \label{lem:OneUmbrellaConstruction}
Let $U_0$ be the root bonnet of a bonnet system with root-edge $(u,v)$. Then there exists a flat visibility representation $\Gamma$ of $U_0$ on three layers such that
\begin{enumerate}
\item $(u,v)$ spans the top layer of $\Gamma$.
\item Any anchor-edge of $U_0$ is drawn horizontally in the middle or bottom layer.
\end{enumerate}
\end{lemma}
\begin{proof}
As a first step, we draw the ribbon of $U_0$
on 2 layers in such a way that $(u,v)$ 
and all anchor-edges are drawn horizontally;
see Figure \ref{fig:ReleaseEdgeExample}(a) for an illustration. 
(This part is identical to \cite{Bie-WAOA12}.)
Consider the standard embedding of $P$ in which the
dual tree is a path, say it consists of faces $f_1,\dots,f_k$.  We draw
$k+1$ vertical edges between two layers, with the goal that the region
between two consecutive ones belong to $f_1,\dots,f_k$ in this order.
Place $u$ and $v$ as segments in the top layer, and with an $x$-range such
that they touch all the regions of faces that $u$ and $v$ are incident to.
Similarly create segments for all other vertices.  The placement for the
vertices is uniquely determined by the standard planar embedding,
except for the vertices incident to $f_1$ and $f_k$. 
We place those vertices
such that the leftmost/rightmost vertical edge is not an anchor-edge.
To see that this is possible, 
recall that $P$ connects two non-cutting edges $e_1,e_2$ of $G$ that are
incident to $f_1$ and $f_k$.  If $e_1\neq (u,v)$, then choose the layer
for the vertices of $f_1$ such that $e_1$ is drawn vertical.
If $e_1=(u,v)$,  then one of its ends (say $u$) is the degree-2 vertex on 
$f_1$ and drawn in the top-left corner.  
The other edge $e'$ incident to $u$ is not an anchor-edge of $U$ by
Observation~\ref{obs:HangingSubgraphLocations}, and we draw $e'$ vertically.
So the leftmost vertical edge is either a non-cutting edge (hence not
an anchor-edge) or edge $e'$ (which is not an anchor-edge).  We proceed
similarly at $f_k$ so that the rightmost vertical edge is not an anchor-edge.
Finally all other vertical edges are cutting edges of $U_0$ and hence not
anchor-edges.

The drawing of $P$ obtained in this first step has $(u,v)$ in the top layer.
As a second step, we now {\em release}
$(u,v)$ as in \cite{Bie-WAOA12}.  This operation can be applied to any
edge that is drawn  horizontally in the top layer of a flat visibility
representation. It consists of adding a layer above
the drawing, moving $(u,v)$ into it, and re-routing edges by expanding
vertical ones at $u$ and $v$, and turning horizontal ones into vertical ones.  
In the result, $(u,v)$ spans the top layer. 
See Figure \ref{fig:ReleaseEdgeExample}(b) for an illustration and \cite{Bie-WAOA12} for details.

\begin{figure}[ht!]
\hspace*{\fill}
\begin{subfigure}[b]{0.25\linewidth}
	\includegraphics[scale=0.5,trim=90 0 80 0,clip,page=1]{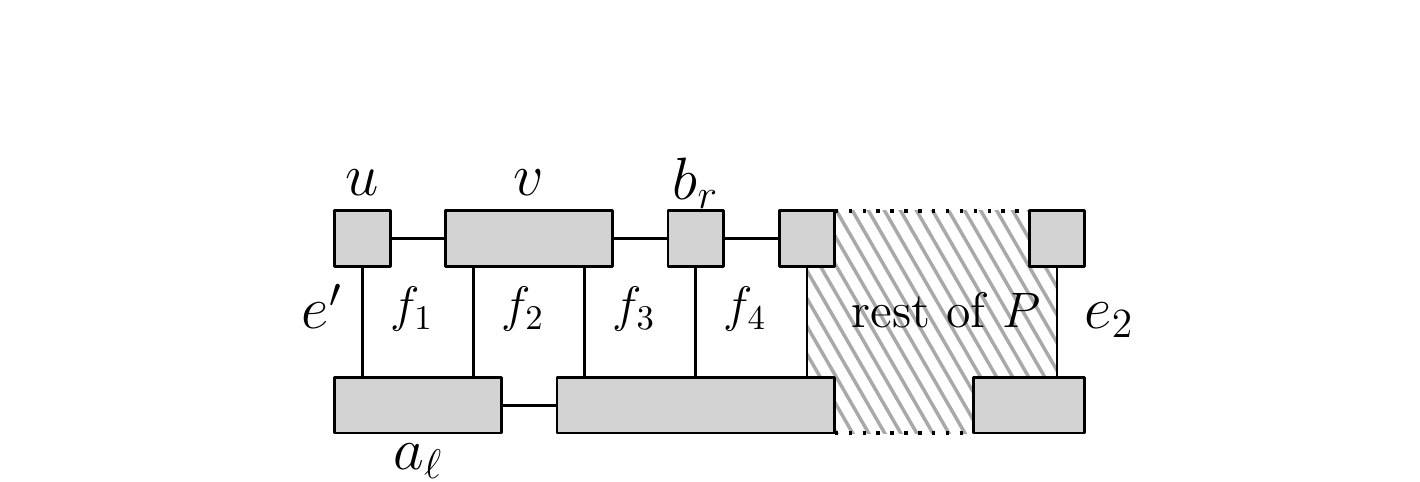}
	\caption{Drawing the ribbon.}
\end{subfigure}
\hspace*{\fill}
\begin{subfigure}[b]{0.25\linewidth}
	\includegraphics[scale=0.5,trim=90 0 80 0,clip,page=2]{ReleaseEdgeExampleFigure.pdf}
	\caption{Releasing $(u,v)$.}
\end{subfigure}
\hspace*{\fill}
\\[2ex]
\begin{subfigure}[b]{0.4\linewidth}
	\includegraphics[scale=0.5,page=3]{ReleaseEdgeExampleFigure.pdf}
	\caption{Adding the fans.}
\end{subfigure}
\hspace*{\fill}
\begin{subfigure}[b]{0.4\linewidth}
	\includegraphics[scale=0.5,page=3]{P2InsertSubsystemsFigure.pdf}
	\caption{Merging hanging subgraphs.}
\end{subfigure}
\hspace*{\fill}
\caption{From bonnet system to drawing.}
\label{fig:ReleaseEdgeExample}
\label{fig:HandleStandardEmbedding}
\label{fig:UmbrellaConstructionBaseCasePart2}
\label{fig:P2InsertSubsystems}
\end{figure} 

We now have a drawing of the ribbon $P$ on 3 layers where $(u,v)$ spans the top,
say $u$ is in the top left corner.  As the third and final step,
we add the remaining vertices of the two fans.  
Consider the fan $F_1$ at $u$, and let $(u,a_\ell)$ be the edge that it
has in common with the ribbon $P$. We have two possible cases.  If
$(u,a_\ell)$ is the leftmost vertical edge, then simply place the 
vertices of the fan to the left of $a_\ell$ and extend $u$ (see 
Figure \ref{fig:UmbrellaConstructionBaseCasePart2}(c)).  Else,
$(u,a_\ell)$ was drawn horizontally in the drawing of the first step (because
it reflects the planar embedding), and therefore after releasing $(u,v)$
there is space to the right of $a_\ell$.  Into this space we insert the
remaining vertices of the fan at $u$.
The fan at $v$ is added in a symmetric fashion.
Figure \ref{fig:UmbrellaConstructionBaseCasePart2}(c) illustrates
the second case for the fan at $v$.

This finishes the construction of the drawing of $U_0$.
It remains to show that all anchor-edges are horizontal and in the bottom two
layers. We ensured
that this is the case in the first step.  Releasing $(u,v)$ adds more vertical 
edges, but all of them are incident to $u$ or $v$ and not anchor-edges by
Observation~\ref{obs:HangingSubgraphLocations}. 
Likewise, all vertical edges added when inserting the fans are incident to $u$ 
or $v$.  The only  horizontal edge in the top layer is $(u,v)$, which is not
an anchor-edge.
This finished the proof of Lemma \ref{lem:OneUmbrellaConstruction}.
\end{proof}

Now we explain how to merge hanging subgraphs.

\begin{theorem} \label{thm:UpperBound}
Any maximal outerplanar graph $G$ has a planar flat visibility representation of height at most $2bd^{\mathit{free}}(G)+1$.
\label{lem:DrawingGivenPathSystem}
\end{theorem}
\begin{proof}
We show by induction that any graph with a bonnet system
$\mathcal{U}$ of depth $H$ has a drawing $\Gamma$ of height $2H+1$ where 
the root-edge $(u,v)$ spans the top layer.  This proves the theorem when
applying it to a bonnet system $\mathcal{U}$ of depth $bd^{\mathit free}(G)$.  

Let $U_0$ be the root bonnet of the bonnet system, and
draw $U_0$ on 3 layers using Lemma \ref{lem:OneUmbrellaConstruction}.
Thus $(u,v)$ spans the top and any anchor-edge $(a,b)$ of $U_0$
is drawn as a horizontal edge in the bottom two layers of $\Gamma_0$.
If $H=1$ then we are done.  Else
add $2H-2$ layers to $\Gamma_0$ between the middle and bottom layers. 
For each anchor-edge $(a,b)$ of $U_0$,
the hanging subgraph $S_{a,b}$ of $U_0$ has a bonnet system of depth at most 
$H-1$ with root-edge $(a,b)$.  By induction $S_{a,b}$ has a drawing $\Gamma_1$ 
on at most $2H - 1$ layers with $(a,b)$ spanning the top layer. 

If $(a,b)$ is in the bottom layer of $\Gamma_0$, then we can rotate (and reflect, if necessary) $\Gamma_1$ so that $(a,b)$ is in the bottom layer of $\Gamma_1$ and the left-to-right order of $a$ and $b$ in $\Gamma_1$ is the same as their left-to-right order in $\Gamma_0$. This updated drawing of $\Gamma_1$ can then be inserted in the space between $(a,b)$ in $\Gamma_0$. This fits because $\Gamma_1$ has height at most $2H-1$, and in the insertion process we can re-use the layer spanned by $(a,b)$.
If $(a,b)$ is in the middle layer of $U_0$, then we can reflect $\Gamma_1$ (if necessary) so that $(a,b)$ has the same left-to-right order in $\Gamma_1$ as in $\Gamma_0$. This updated drawing of $\Gamma_1$ can then be inserted in the space between $(a,b)$ in $\Gamma_0$.
See Figure \ref{fig:P2InsertSubsystems}(d).
Since we added $2H-2$ layers to a drawing of height 3, the total height of the final drawing is $2H + 1$ as desired.
\end{proof}

Our proof is algorithmic, and finds a drawing, given a bonnet system,
in linear time.  One can also show (see the appendix) that the rooted bonnet
depth, and an associated bonnet system, can be found in linear time
using dynamic programming.  Hence the run-time to find this drawing
is linear.

\medskip\noindent{\bf Comparison to \cite{Bie-WAOA12}: }
The algorithm in \cite{Bie-WAOA12} has only two small difference.  The main
one is that it does not do the ``third step'' when drawing the root umbrella,
thus it draws the ribbon but not the fans.  Thus in the induction step our
algorithm always draws at least as much as the one in \cite{Bie-WAOA12}. 
Secondly, \cite{Bie-WAOA12} uses a special construction
if $pw(G)=1$ to save a constant number of levels.  This could easily be 
done for our algorithm as well in the case where $pw(G)=1$ but $bd(G)=2$.
As such, our construction never has worse height (and frequently it is better).

\medskip\noindent{\bf Comparison to \cite{Bie-GD02}: }  One can argue
that $bd(G)\leq \log(n+1)$ (see the appendix).  Since \cite{Bie-GD02}
uses $3\log n-1$ levels while ours uses $2bd(G)+1\leq 2\log(n+1)+1$ levels,
the upper bound on the height is better for $n\geq 9$.

\section{From Drawing to Umbrella System}\label{chap:LowerBound}

We now argue that any flat visibility representation of height $H$ gives rise 
to an umbrella system of depth at most $H-1$, proving a lower bound.  
We first briefly sketch the idea.  We assume that we have a drawing such
that for some non-cutting edge $(u,v)$ we have an ``escape path'', i.e., 
a poly-line
to the outerface that does not intersect the drawing. 
Now find an outerplanar path that connects
the leftmost vertical edge $(x,y)$ of the drawing with $(u,v)$.  This
becomes the handle of an umbrella $U$ with cap $(u,v)$, and the fans consist
of all remaining neighbours of $u$ and $v$.  One can now argue that any
hanging subgraph of $U$ is drawn with height at most $H-1$, and furthermore, has
an escape path from its anchor-edge.  The claim then holds by induction.


We first must clarify some definitions.  
Let $\Gamma$ be a flat visibility representation, and let $B_{\Gamma}$ be a minimum-height bounding box of $\Gamma$. 
A vertex $v \in G$ has a \emph{left escape path} in $\Gamma$ if there exists a polyline inside $B_{\Gamma}$ from $v$ to a point on the left side of $B_{\Gamma}$ that is vertex-disjoint from $\Gamma$ except at $v$, and for which all bends are on layers. We say that $(\ell_1,\ell_2)$ is a \emph{left-free edge} of $\Gamma$ if it is vertical, and any layer intersected by $(\ell_1,\ell_2)$ is empty
to the left of $(\ell_1,\ell_2)$.  In particular, both $\ell_1$ and $\ell_2$ have a left escape path by going leftwards in their respective layers.
Define \emph{right escape paths} and \emph{right-free edges} symmetrically; we use {\em escape path} for either a left escape path or a right escape path. See 
Figure \ref{fig:LeftmostRightmostExample}(a). 

\begin{figure}[ht!]
\hspace*{\fill}
\begin{subfigure}[b]{0.4\linewidth}
	\includegraphics[scale=0.45,page=1]{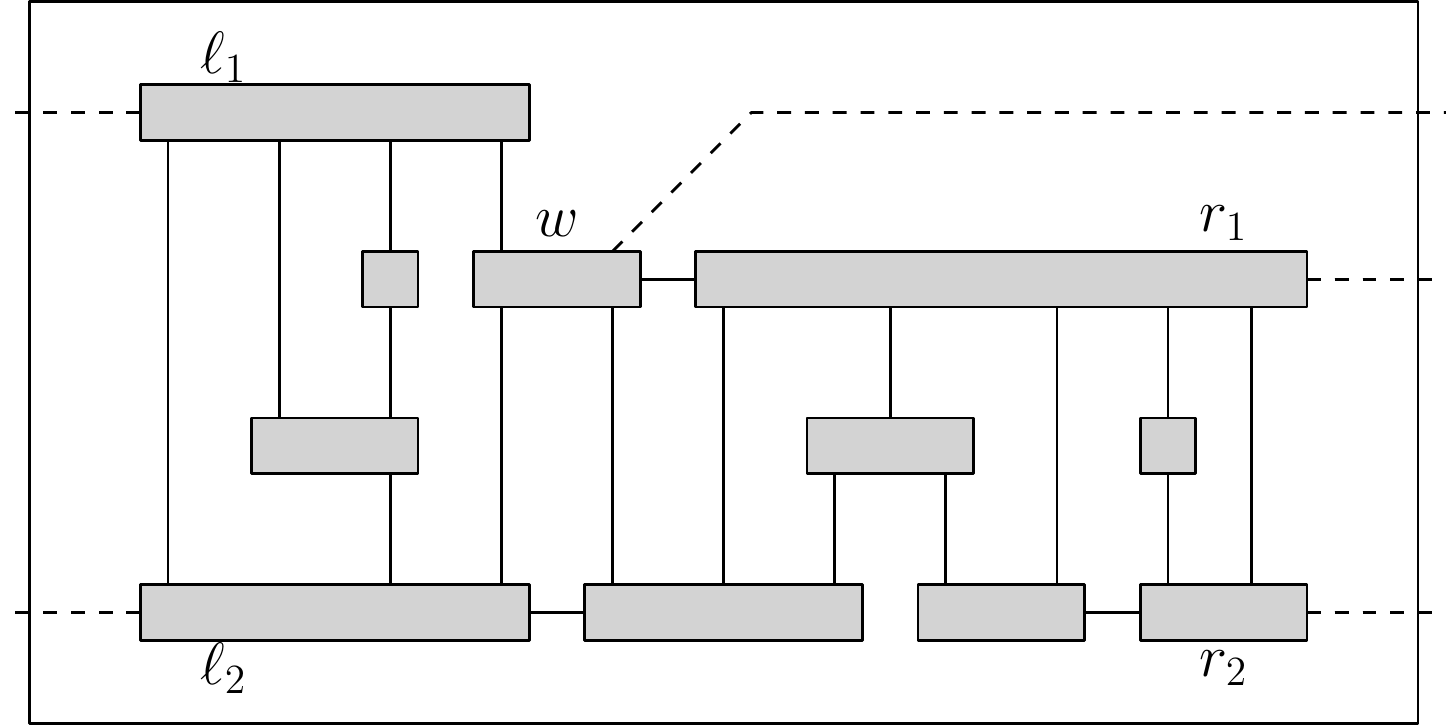}
	\caption{}
\end{subfigure}
\hspace*{\fill}
\begin{subfigure}[b]{0.55\linewidth}
	\includegraphics[scale=0.45,page=2]{LeftmostRightmostExampleFigure.pdf}
	\caption{}
\end{subfigure}
\hspace*{\fill}
\caption{Vertex $w$ has a right escape path, $(\ell_1, \ell_2)$ is left-free 
and $(r_1,r_2)$ is right-free.  After flipping the cutting component at $(\ell_1,\ell_2)$, there is a non-cutting edge that is left-free.}
\label{fig:LeftmostRightmostExample}
\label{fig:NonCuttingLeftmostRightmost}
\end{figure}

It is easy to see that any flat visibility representation has a left-free
edge,  presuming the graph has minimum degree at least 2. 
Let $(v,w)$ be the leftmost vertical edge (breaking ties arbitrarily);
there must be such an edge because the leftmost vertex in each layer has
at most one horizontal incident edge, and therefore at least one vertical one.
In any layer spanned by $(v,w)$, no vertical edge is farther left by
assumption. No vertex can be farther left either, else the incident
vertical edge of the leftmost of them would be farther left.
So $(v,w)$ is left-free.

For the proof of the lower bound, we use as handle an outerplanar path 
connecting to a left-free edge.  Recall that the definition of handle 
requires that it connects to a non-cutting edge, so we need a left-free edge
that is not a cutting edge.  This does not exist in all drawings (see
e.g. Figure~\ref{fig:LeftmostRightmostExample}(a)),
but as we show now, we can modify the drawing without increasing height such
that such an edge exists.  To be able to apply it later, we must also show
that this modification does not destroy a given escape path.

\begin{lemma} \label{lem:LeftmostNonCuttingEdgeWithRightmost}
Let $\Gamma$ be a flat visibility representation of a maximal outerplanar graph $G$. 
\begin{enumerate}
\item Let $(r_1, r_2)$ be a right-free edge of $\Gamma$, and let $w$ be a vertex that has a right escape path. Then there exists a drawing $\Gamma'$ in which $w$ has a right escape path, $(r_1,r_2)$ is a right-free edge, and there exists a left-free edge that is not a cutting edge of $G$.
\item Let $(\ell_1, \ell_2)$ be a left-free edge of $\Gamma$, and let $w$ be a vertex that has a left escape path. Then there exists a drawing $\Gamma'$ in which $w$ has a left escape path, $(\ell_1,\ell_2)$ is a left-free edge, and there exists a right-free edge that is not a cutting edge of $G$.
\end{enumerate}
In either case, the $y$-coordinates of all vertices in $\Gamma$ are unchanged in $\Gamma'$, and in particular both drawings have the same height.
\end{lemma}

\begin{proof}
We prove the claim by induction on $n$ and show only the first claim (the 
other is symmetric).  Let $(\ell_1, \ell_2)$ be the leftmost vertical edge 
of $\Gamma$; this is left-free as argued above.
If $(\ell_1,\ell_2)$ is not a cutting edge of $G$, then we are done with
$\Gamma'=\Gamma$.   In particular,
this is the case if $n=3$ when $G$ is a triangle and has no cutting edge.

So assume $n\geq 4$ and
$(\ell_1,\ell_2)$ is a cutting edge of $G$. 
Let $A$ and $B$ be the cut-components of $(\ell_1,\ell_2)$, named such 
that $w \in A$.
Let $\Gamma_A$ [resp.~$\Gamma_B$] be the drawing of $A$ [$B$] 
induced by $\Gamma$.
Edge $(\ell_1,\ell_2)$ is left-free for both $\Gamma_A$ and $\Gamma_B$.
Reflect $\Gamma_B$ horizontally (this makes $(\ell_1,\ell_2)$ right-free)
to obtain $\Gamma_B'$.
By induction, we can create a drawing $\Gamma''_B$ from $\Gamma'_B$ in which $(\ell_1,\ell_2)$ is right-free and there is a left-free edge $(\ell'_1,\ell'_2)$ that is not a cutting edge of $B$.   We have $(\ell'_1,\ell'_2)\neq (\ell_1,\ell_2)$, because the common neighbour of $\ell_1,\ell_2$ in $B$ forces a vertex or edge to reside to the left of the right-free edge $(\ell_1,\ell_2)$.  So
$(\ell_1',\ell_2')$ is not a cutting edge of $G$ either.

Create a new drawing that places $\Gamma''_B$ to the left of $\Gamma_A$ and extends $\ell_1$ and
$\ell_2$ to join the two copies; this is possible since $(\ell_1,\ell_2)$ has the same $y$-coordinates in $\Gamma_A,\Gamma,\Gamma_B$ and $\Gamma_B''$, and it is left-free in $\Gamma_A$ and right-free in $\Gamma''_B$.  Also delete one
copy of $(\ell_1,\ell_2)$.
See Figure \ref{fig:NonCuttingLeftmostRightmost}(b). 
The drawing $\Gamma_A$ is unchanged, so $w$ will have the same right escape path in $\Gamma'$ as in $\Gamma$, and $\Gamma'$ will have right-free edge $(r_1,r_2)$ and left-free non-cutting edge $(\ell'_1,\ell'_2)$, as desired. 
\end{proof}


We are now ready to prove the lower bound if there is an escape path.

\begin{lemma} \label{lem:DrawingToPathSystem}
Let $\Gamma$ be a flat visibility representation of a maximal outerplanar graph $G$ with height $H$, and let $(u,v)$ be a non-cutting edge of $G$. If there exists an escape path from $u$ or $v$ in $\Gamma$, then $G$ has an umbrella system with root-edge $(u,v)$ and depth at most $H-1$.
\end{lemma}
\begin{proof}
We proceed by induction on $H$. 
Assume without loss of generality that there exists a right escape path from $v$ (all other cases are symmetric). 
Using Lemma \ref{lem:LeftmostNonCuttingEdgeWithRightmost}, we can modify $\Gamma$ without increasing the height so that $v$ has a right escape path, and there 
is a left-free edge $(\ell_1,\ell_2)$ 
in $\Gamma$ that is a not a cutting edge of $G$. 
Let $P$ be the outerplanar path that connects edge $(\ell_1,\ell_2)$ and $(u,v)$.
Let $U_0$ be the union of $P$, the neighbors of $u$, and the neighbors of $v$;
we use $U_0$ as the root umbrella of an umbrella system. 

We now must argue that all hanging subgraphs of $U_0$ are drawn with height
at most $H-1$ and have escape paths from their anchor-edges; we can then find 
umbrella systems for them by induction and combining them with $U_0$ gives
the umbrella system for $G$ as desired.  

To prove the height-bound, we define ``dividing paths''
as follows.   The outer-face of $U_0$ in the standard embedding contains
$(\ell_1,\ell_2)$ (since it is not a cutting edge) as well as $v$.  Let
$P_1$ and $P_2$ be the two paths from $\ell_1$ and $\ell_2$ to $v$ along
this outer-face in the standard embedding.  
Define the {\em dividing path} $\Pi_i$ (for $i=1,2$) to be the
poly-line in $\Gamma$ that consist
of the left escape path from $\ell_i$, then the drawing of the path $P_i$
(i.e., the vertical segments of its edges and parts of the horizontal
segments of its vertices),
and then the right escape path from $v$.  
See Figure \ref{fig:DrawingToPathSystem}.  

\begin{figure}[ht!]
\hspace*{\fill}
\includegraphics[scale=0.5,page=2]{DrawingToPathSystemFigure.pdf}
\hspace*{\fill}
\includegraphics[scale=0.5,page=1]{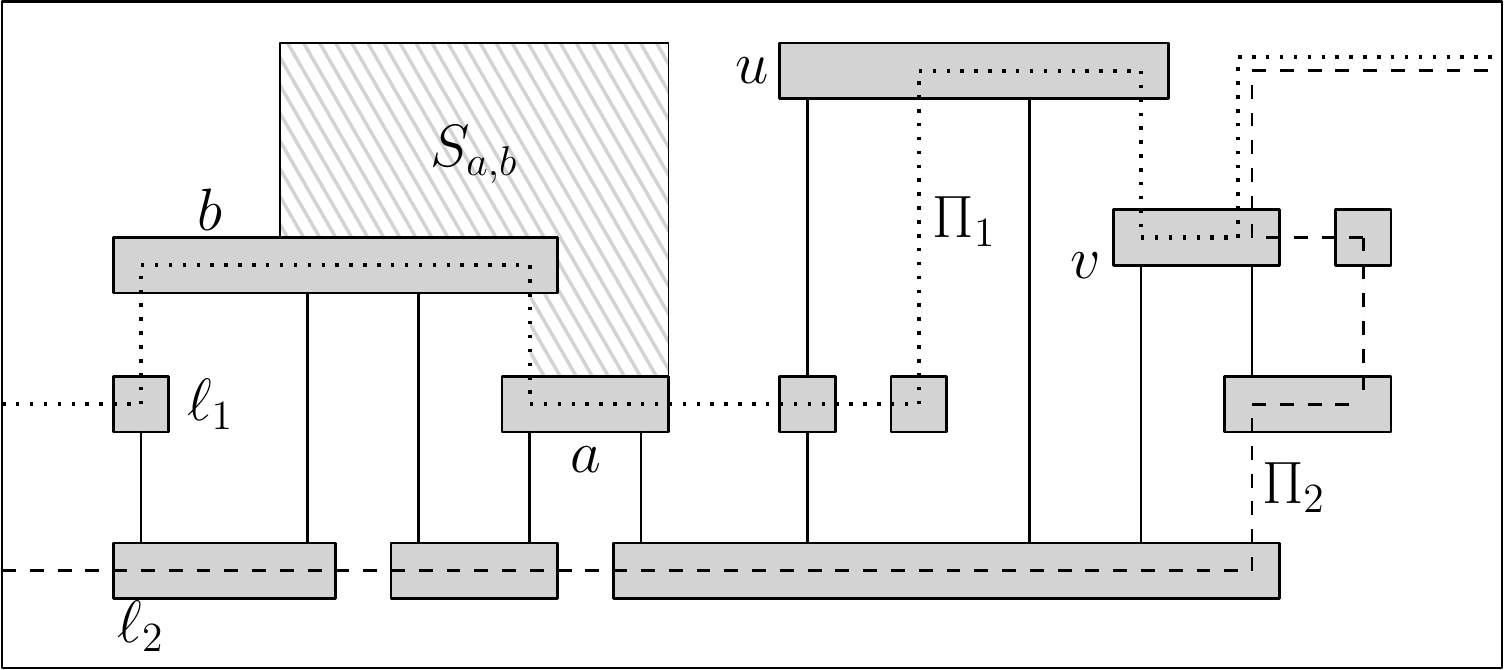}
\hspace*{\fill}
\caption{Extracting dividing paths from a flat visibility representation. $P_1$/$\Pi_1$ is dotted while $P_2/\Pi_2$ is dashed. \label{fig:DrawingToPathSystem}}
\end{figure}

Now consider any
hanging subgraph $S_{a,b}$ of $U_0$ with anchor-edge $(a,b)$. 
No edge incident to $v$ is an anchor-edge, and neither is $(\ell_1,\ell_2)$,
since it is not a cutting edge.  So $(a,b)$ is an edge of $P_1$ or $P_2$ (say
$P_1$) that is not incident to $v$.  Therefore $(a,b)$ (and with it $S_{a,b}$)
is vertex-disjoint from $P_2$.
In consequence, the drawing $\Gamma_S$ of $S_{a,b}$ induced by $\Gamma$ is
disjoint from the dividing path $\Pi_2$.  Since $\Pi_2$ connects a point
on the left boundary with a point on the right boundary, therefore
$\Gamma_S$ must be entirely above or entirely below $\Pi_2$, 
say it is above.  
Since
$\Pi_2$ has all bends at points with integral $y$-coordinate,
therefore the bottom layer of $\Gamma$ is not available for $\Gamma_S$.
In consequence $\Gamma_S$ has height at most $H-1$ as desired.

Recall that $(a,b)$ belongs to $P_1$ and is not incident to $v$.
After possible renaming of $a$ and $b$, we may assume that $b$ is closer to
$\ell_1$ along $P_1$ than $a$.  
Then the sub-path of $P_1$ from $b$ to $\ell_1$ is
interior-disjoint from $S_{a,b}$.  This path, together with the left escape path
from $\ell_1$, is a left escape path from $b$ that resides within the top
$H-1$ layers, because it does not contain $v$ and hence is disjoint from $\Pi_2$.
We can hence apply induction to $S_{a,b}$ to obtain an umbrella system of
depth at most $H-2$ with root-edge $(a,b)$.  Repeating this for all 
hanging subgraphs, and combining the resulting umbrella systems with
$U_0$, gives the result.
\end{proof}


\begin{theorem} \label{thm:LowerBound}
Let $G$ be a maximal outerplanar graph.  If $G$
has a flat visibility representation $\Gamma$ of height $H$, 
then $ud^{\mathit{free}}(G)\leq H-1$.
\end{theorem}

\begin{proof}
Using Lemma \ref{lem:LeftmostNonCuttingEdgeWithRightmost}, we can convert $\Gamma$ into a drawing $\Gamma'$ of the same height in which some edge $(u,v)$ is a right-free non-cutting edge. This implies that there is a right escape path from $v$, and by Lemma~\ref{lem:DrawingToPathSystem} we can find an umbrella system of $G$ with root-edge $(u,v)$ and depth $H-1$.  So $ud^{\mathit{free}}\leq ud(G;u,r)\leq H-1$.
\end{proof}

\section{Conclusions and Future Work}\label{chap:Conclusions}

We presented an algorithm for drawing maximal outerplanar graphs that is a 2-approximation for the optimal height. To this end, we introduced the umbrella depth as a new graph parameter for maximal outerplanar graphs, and used as key result that any drawing of height $H$ implies an umbrella-depth of at least $H-1$.
Our result significantly improves the previous best result, which was based on 
the pathwidth and gave a 4-approximation.
We close with some open problems:
\begin{itemize}
\item Our result only holds for maximal outerplanar graphs. Can the algorithm be modified so that it works for all outerplanar graphs?  Specifically, can we make an outerplanar graph maximal in such a way that the umbrella depth does not increase (much)?
\item The algorithm from Section~\ref{chap:UpperBound} creates a drawing that does not place all vertices on the outer face. Can we create an algorithm that minimizes or approximates the optimal height when the standard planar embedding must be respected?
\item What is the width achieved by the algorithm from Section~\ref{chap:UpperBound}? Any visibility representation can be modified without changing the height so that the width is at most $m+n$, where $m$ is the number of edges and $n$ is the number of vertices \cite{Bie-GD14}. Thus the width is $O(n)$, but what is the constant?
\item Is it possible to determine the optimal height for maximal outerplanar graphs in polynomial time?
\end{itemize}
Finally, are there approximation algorithms for the height or the area of drawings for other, more general planar graph classes? 


\bibliographystyle{plain}
\bibliography{journal,full,gd,papers}

\newpage
\begin{appendix}
\section{Computing the Depth}\label{chap:Algorithm}

We now introduce a dynamic programming algorithm for finding the rooted 
bonnet depth of a maximal outerplanar graph $G$, relative to a given
fixed root-edge $(u,v)$.  (A very similar algorithm finds the umbrella
depth; we leave those details to the reader.)

As before, for any cutting edge $(a,b)$ of $G$ let $S_{a,b}$
is the cut-component of $G$ that does not contain $(u,v)$.
It will also be convenient to define $S_{a,b}:=\{(a,b)\}$ if $(a,b)$ is
a non-cutting edge with $(a,b)\neq (u,v)$, and to define $S_{u,v}:=G$.
We also use the notation $ud(a,b):=ud(S_{a,b};a,b)$ and define $ud(a,b)=0$
if $S_{a,b}$ is a single edge.

We first sketch the overall idea.  Consider one subgraph $S_{a,b}$ and
the root bonnet $U_0$ used in the bonnet system of minimum depth.
The ribbon of $U_0$ is an outerplanar path connecting
edges $e_1$ and $e_2$, but we can view it as two outerplanar
paths, connecting the root-edge $(a,b)$ to one of $e_1$ and $e_2$.  If $c$
is the common neighbour of $a$ and $b$, then we can hence split the bonnet 
into two parts by removing the face $\{a,b,c\}$.  Each part looks much like
an umbrella, except that they are rooted at $(a,c)$ and $(c,b)$, respectively,
and each of them has no fan at $c$.  
See also Figure~\ref{fig:BonnetDepthFormula}(a).
Therefore the minimum bonnet
depth can be found by finding the minimum-depth in the two 
subgraphs $S_{a,c}$ and $S_{b,c}$, under the restriction that the
root bonnets in the corresponding systems are such partial umbrellas.
We must repeat the argument for the partial umbrellas, breaking them 
into a fan and a handle, and hence end up computing 6 different types
of depths for each anchor edge, but 6 types are enough and hence the
overall run-time is linear.
We now define these variants of bonnet-depth for $S_{a,b}$:
\begin{itemize}
\item A \emph{handle} is an outerplanar path that connects $(a,b)$ to some
non-cutting edge.
See Figure~\ref{fig:BonnetDepthFormula}(b).
Define $$bd^{h}(a,b) := 1+ \min_{U^h} \Big\{ \max_{(x,y)} \; bd(x,y) \Big\}$$
where $U^h$ is a handle and $(x,y)$ is a non-cutting edge of $U^h$. 
\item The \emph{full fan at $a$} [resp.~\emph{full fan at $b$}] consists of $a$ [$b$] and all neighbours of $a$ [$b$].  
See Figure~\ref{fig:BonnetDepthFormula}(c).
Define $$bd^{f_a}(a,b) := 1+ \max_{(x,y)} \; bd(x,y)$$
where $(x,y)$ is a non-cutting edge of the full fan at $a$.
Symmetrically define $bd^{f_b}(a,b)$ using the full fan at $b$.
\item A \emph{partial-$a$ umbrella} [resp.~\emph{partial-$b$ umbrella}] is an umbrella with cap $(a,b)$ that contains all neighbours of $a$ [$b$] and in which the fan at $b$ [$a$] is empty.
See Figure~\ref{fig:BonnetDepthFormula}(d).
Define $$bd^{p_a}(a,b) := 1+ \min_{U^{p_a}} \Big\{ \max_{(x,y)} \; bd(x,y) \Big\}$$
where $U^{p_a}$ is a partial-$a$ umbrella and $(x,y)$ is a non-cutting edge of $U^{p_a}$.
Symmetrically define $bd^{h_b}(a,b)$ using a partial-$b$ umbrella.
\end{itemize}

\begin{figure}[ht!]
\hspace*{\fill}
\begin{subfigure}[b]{0.3\linewidth}
	\includegraphics[scale=0.7,page=4]{SplitBonnets.pdf}
	\caption{Bonnet.}
\end{subfigure}
\hspace*{\fill}
\begin{subfigure}[b]{0.6\linewidth}
	\includegraphics[scale=0.7,page=1,trim=0 0 20 0,clip]{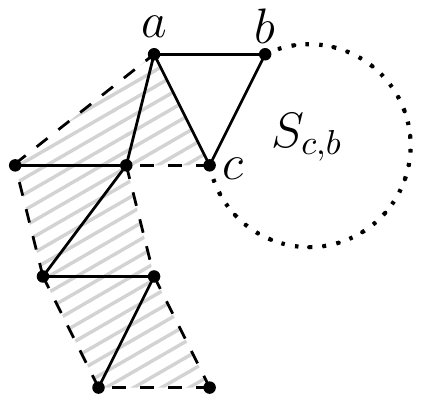}
	\includegraphics[scale=0.7,page=6,trim=50 0 0 0,clip]{SplitBonnets.pdf}
	\caption{Handle.}
\end{subfigure}
\hspace*{\fill}
\newline
\hspace*{\fill}
\hspace*{\fill}
\begin{subfigure}[b]{0.15\linewidth}
	\includegraphics[scale=0.7,page=3,trim=50 40 20 0,clip]{SplitBonnets.pdf}
	\caption{Full fan.}
\end{subfigure}
\hspace*{\fill}
\begin{subfigure}[b]{0.6\linewidth}
	\hspace*{\fill}
	\includegraphics[scale=0.7,page=2,trim=0 0 20 0,clip]{SplitBonnets.pdf}
	\includegraphics[scale=0.6,page=5,trim=50 0 0 0,clip]{SplitBonnets.pdf}
	\hspace*{\fill}
	\caption{Partial umbrella.}
\end{subfigure}
\hspace*{\fill}
\caption{For the recursive formulas.  Edges that may be anchor-edges are dashed.}
\label{fig:BonnetDepthFormula}.
\end{figure}

\begin{lemma}
\label{lem:BonnetDepthFormula}
We have $ud^t(a,b)=0$ if $S_{a,b}$ is a single edge, where $t\in \{\emptyset,
h,f_a,f_b,p_a,p_b\}$.  Else, letting $c$ be the common neighbour of $a$ and $b$,
we have
\begin{enumerate}
\item $bd(a,b) = \max\{ bd^{p_a}(a,c), bd^{p_b}(c,b) \}$
\item $bd^{h}(a,b) = \min\left[ \max\{ bd^{h}(a,c), 1+bd(c,b) \},
			\max\{ 1+bd(a,c), bd^h(c,b) \} \right]$
\item $bd^{f_a}(a,b) = \max\{ bd^{f_a}(a,c), 1+bd(c,b) \}$
\item $bd^{f_b}(a,b) = \max\{ 1+bd(a,c), bd^{f_b}(c,b) \}$
\item $bd^{p_a}(a,b) = \min\left[ \max\{ bd^{p_a}(a,c), 1+bd(c,b) \},
			\max\{ bd^{f_a}(a,c), bd^h(c,b) \} \right]$
\item $bd^{p_b}(a,b) = \min\left[ \max\{ 1+bd(a,c), bd^{p_b}(c,b) \},
			\max\{ bd^{h}(a,c), bd^{f_b}(c,b) \} \right]$
\end{enumerate}
\end{lemma}
\begin{proof}
All these formulas are proved in a similar way:  Consider the root bonnet $U_0$
of the bonnet system that achieves the depth on the left side.  
When splitting this bonnet into two by removing the face 
$\{a,b,c\}$, we obtain two bonnets for the two subgraphs $S_{a,c}$ and
$S_{c,b}$ and can argue what type they are.  (Sometimes there may be two
possibilities, depending on which direction the handle of $U_0$ went, in case
of which the one that yields the smaller depth is used.)  This proves 
``$\geq$'', and equality is easily shown by putting together bonnet systems
of $S_{a,c}$ and $S_{c,b}$ of the appropriate type.

We demonstrate this in detail for (5); see also Figure~\ref{fig:BonnetDepthFormula}(d).
So assume we have an umbrella system $\mathcal{U}$ with root-edge 
$(a,b)$ that has depth $bd^{p_a}(a,b)$ and where the root-bonnet $U_0$ is
a partial-$a$ umbrella.  So $U_0$ consists of a handle $P$ and a fan at $a$
that includes all neighbours of $a$ not in $P$.  Since $P$ is a handle,
it connects to $(a,b)$, therefore
not both edges $(a,c)$ and $(c,b)$ can be cutting edges of $P$.  

We distinguish cases.  In the first
case, $(c,b)$ is not a cutting edge of $P$.  
Since the fan at $b$ is empty in $U_0$, $(c,b)$ is not a
cutting edge of $U_0$ either.
Thus the formula for $bd^{p_a}(a,b)$ includes the
term $1+ud(c,b)$ in the maximum.    
Furthermore, $U_0-\{b\}$ is a partial-$a$ umbrella for
$S_{a,c}$, and this, together with the subsystems of $\mathcal{U}$ for
its hanging subgraphs, gives
a bonnet system with depth at least $bd^{p_a}(a,c)$.  Therefore
$bd^{p_a}(a,b) \geq \max\{bd^{p_a}(a,c),1+ud(c,b)\}$ in the first case.

In the second case, $(a,c)$ is not a cutting edge of $P$.  Therefore the
fan of $U_0$ at $a$ consists of $a$ and all neighbours of $a$ except $b$,
and thus is the full fan $F$ at $a$ in the subgraph $S_{a,c}$.  The subsystems
of $\mathcal{U}$ for hanging subgraphs of $F$
give a bonnet system for $S_{a,c}$
where the root bonnet is the full fan at $a$, hence it has depth at least
$bd^{f_a}(a,c)$.  Furthermore, $P-\{a\}$ is a handle for $S_{c,b}$.
So $\mathcal{U}$ implies
a bonnet system for $S_{c,b}$ where the root is a handle, hence it has
depth at least $bd^{h}(c,b)$.
Therefore $bd^{p_a}(a,b) \geq \max\{bd^{f_a}(a,c),ud^h(c,b)\}$ in the
second case.  One of the two cases must apply, and so $bd^{p_a}(a,b)$
is at least as big as the smaller of the two bounds and ``$\geq$'' holds.

To show ``$\leq$'', let us assume that the minimum is achieved at 
$\max\{bd^{f_a}(a,c),ud^h(c,b)\}$ (the other case is similar).
Find a bonnet system $\mathcal{U}^f$ of $S_{a,c}$ of depth $bd^{f_a}(a,c)$
where the root
bonnet $F$ is the full fan at $a$, and a bonnet system $\mathcal{U}^h$
of $S_{c,b}$ of depth $ud^h(c,b)$
where the root bonnet $U^h$ is a handle.   Then
$F\cup U^h \cup \{(a,b)\}$ is a partial-$a$ umbrella $U_0$,
and combining it with the subsystems of $\mathcal{U}^f$ and $\mathcal{U}^h$
gives a bonnet system of $S_{a,c}$ whose depth is
$\max\{bd^{f_a}(a,c),ud^h(c,b)\}$.
\end{proof}

We can convert these formulas into a 
dynamic programming algorithm to compute the bonnet depth by using 
the standard 
bottom-up traversal in a tree. Given a maximal outerplanar graph $G$,
initialize $ud^t(a,b)=0$ for all types $t$ and for all 
non-cutting-edges $(a,b)$ with the exception of $(u,v)$.
Root the dual tree $T$ at the face incident to root-edge $(u,v)$. 
Any node of $T$ is associated with a cutting-edge of $G$ by taking the 
dual of the arc that connects the node with its parent in $T$.
Traversing $T$ bottom up, when we encounter a node $f$ of $T$ (hence a face
of $G$) we have obtained the bonnet depth values for two out of the
three edges incident to $f$ already, and can compute the bonnet depth
values for the third using the above formulas.  This takes $O(1)$ time
since there are $6$ values and each formula can be evaluated in constant
time.  Finally we evaluate at the root of $T$, which gives 
$bd(G;u,v)$.  Since $T$ has $n-3$ nodes, the
total run-time is $O(n)$.

\begin{theorem} \label{thm:FindingUmbrellaDepth}
Given a non-cutting edge $(u,v)$, there exists an $O(n)$ algorithm to find the rooted umbrella depth $ud(G;u,v)$ of a maximal outerplanar graph $G$ with $n$ vertices. 
\end{theorem}

\subsection{Free vs. Rooted Umbrella/Bonnet Depth}

Note that our algorithm computes the {\em rooted} bonnet depth for $G$, since the root-edge $(u,v)$ must be given. One way to instead find the free bonnet depth is to repeat the process described above for every choice of root-edge in $G$. This would give an $O(n^2)$ algorithm for finding the free bonnet depth. One could likely compute the free bonnet depth in $O(n)$ time by initializing $ud^t(a,b)=0$ for {\em all} non-cutting edges, then updating at the face where the resulting bonnet depth is minimized, and using as root-edge one near where we update last. However, 
as we will show now, the free bonnet depth is at most one less than the rooted bonnet depth, and therefore it does not seem worth the minor improvement to work out the details of this approach.

\begin{lemma} \label{lem:FreeVsRootedUmbrellaDepth}
Given a maximal outerplanar graph $G$, we have 
$$ bd^{\mathit{free}}(G) = \min_{(u,v)} \left\{ bd(G,u,v) \right\} \le \max_{(u,v)} \left\{ bd(G,u,v) \right\} \le bd^{\mathit{free}}(G)+1 $$
where the minimum and maximum are taken over all non-cutting edges $(u,v)$ of $G$.
\end{lemma}

\begin{proof}
The first equality holds per definition, and the second inequality is obvious, so we focus on the third inequality.
Let $\mathcal{U}^*$ be a bonnet system on $G$ with depth $H := bd^{\mathit{free}}(G)$ and let $(u^*,v^*)$ be its root-edge, which by definition is not a cutting edge.   Let $(u,v)$ be an arbitrary non-cutting edge;
it suffices to show $bd(G;u,v)\leq bd(G;u^*,v^*)+1=H+1$.
Let $P$ be the outer-planar path that connects $(u,v)$
and $(u^*,v^*)$ and define a bonnet $U$ with cap $(u,v)$ to consist
of $P$ and the fans at $u$ and $v$ that include all neighbours of $u,v$ not
in $P$.\footnote{$U$ is actually an umbrella, and indeed the same chain of 
inequalities holds if we replace `$bd$' by `$ud$' everywhere.}

We claim that any hanging subgraph $S_{a,b}$ of $U$ with anchor-edge $(a,b)$
has rooted bonnet depth $bd(S_{a,b};a,b)\leq H$.  
Observe that $(u^*,v^*)$ is not an edge of $S_{a,b}$, because 
$(u^*,v^*)$ is not a cutting edge, and
$(u^*,v^*)\in U$ while $S_{a,b}$ is disjoint from $U$ except at
cutting edge $(a,b)$.
Therefore the cutting edge $(a,b)$
has the root-edge $(u^*,v^*)$ of $\mathcal{U}^*$ in one component and 
$S_{a,b}$ in the other.  
One easily argues that therefore $bd(S_{a,b};a,b)\leq d(\mathcal{U}^*)$,
because the bonnets of $\mathcal{U}^*$ 
can be used to build a bonnet system $\mathcal{U}_S$ of $S_{a,b}$
after trimming parts in $G-S_{a,b}$ and expanding each bonnet
as to include all neighbours of the ends of its cap.
%
Thus $S_{a,b}$ has an umbrella system
$\mathcal{U}_S$ with depth at most $d(\mathcal{U}^*)=H$.  
Combining the umbrella systems of these hanging subgraphs with
$U$ gives an umbrella system of $G$ with root-edge $(u,v)$ and depth at
most $H+1$ as desired.
\end{proof}

\section{Comparison with Other Graph Parameters}\label{chap:Comparison}

In this section, we compare the bonnet depth and the umbrella depth to
the pathwidth (the other graph parameter previously used for graph
drawing purposes), as well as the so-called rooted pathwidth.
We define these parameters first.  Let $T$ be a tree.  The {\em pathwidth}
$pw(T)$ of $T$ is 0 if $T$ is a single node, and $1+ \min_P \max_{T'\subset T-P} pw(T')$ otherwise.  Here $P$ is an arbitrary path in $T$, and $T'$ is any
subtree that remains after deleting the vertices of $P$.  Any path where the minimum is achieved
is called a {\em main path} of $T$.   See \cite{Sud04} for more details.
The rooted pathwidth
is quite similar, but forces the path to end at the root. 
Thus,
let $T$ be a rooted tree.  The {\em rooted pathwidth} $rpw(T)$ is 0 
if $T$ is empty
and $1+\min_{P_r} \max_{T'\subset T-P_r} rpw(T')$ otherwise.
Here $P_r$ is a path in $T$ that ends at the root of $T$ and $T'$ is any
subtree that remains after deleting the vertices of $P_r$, where $T'$ is 
rooted as induced by the root of $T$.
Any path where the minimum is achieved is called an {\em rpw-main path}.
See \cite{Bie-OPTI-ArXiV} for more details.
We write $rpw(T,r)$ if the root $r$ is not clear from the context, and define
$rpw^{\mathit{free}}(T) := \min_r rpw(T,r)$.   It is not hard to see that
$\min_r rpw(T,r)$ is attained at a leaf $r$,
because for any interior node $r$ we could have used an even longer rpw-main
path to reach a leaf without making any subtree bigger.

\begin{lemma} 
\label{lem:RootedPathwidthComparison}
Let $G$ be a maximal outerplanar graph $G$ with dual tree $T$,.
Then
$$ \frac{1}{2} pw(T) \le bd^{\mathit{free}}(G) \le ud^{\mathit{free}}(G) 
\le rpw^{\mathit{free}}(T) \leq \max\{1,2pw(T)\}.$$
\end{lemma}
\begin{proof}
We first show that $pw(T)\leq 2bd(G;u,v)$ for any choice of root-edge $(u,v)$.
Fix any bonnet system that has depth $H:=bd(G;u,v)$ and root-edge $(u,v)$, and
let $U_0$ be its root-umbrella.  Recall that $U_0$ is split into the ribbon
$P$, which is an outer-planar path, and the two fans $F_1,F_2$, which are also 
outer-planar paths.  Use the dual tree $P^*$ of $P$ as main path for $T$, and 
for $i=1,2$, use the dual tree $F_i^*$ of $F_i$ in the subtree $T_i$ of $T-P^*$ 
that contains $F_i^*$.  Any subtree $T'$ of $T-P^*-F_1^*-F_2^*$ corresponds
to a hanging subgraph of $U$ that has bonnet depth at
most $H-1$.  By induction $pw(T')\leq 2H-2$,
and $pw(T)\leq 2H$ as desired.

Since an umbrella is a bonnet, we have $bd(G;u,v)\leq ud(G;u,v)$ for all
non-cutting edges $(u,v)$, and the second inequality holds.

For the third inequality, assume that
$H:=rpw^{\mathit{free}}(T)$
is attained when rooting $T$ at leaf $r$, 
and let $(u,v)$ be a non-cutting edge of $G$ incident to $r$.
We claim that $ud(G;u,v)\leq rpw(T,r)=H$.
Let $P^*$ be an rpw-main path of $T$; without loss of generality we
may assume that $P^*$ connects from $r$ to a leaf of $T$ (else a longer
path could be used). Let $P$ be the 
outerplanar path whose dual tree is $P^*$; it connects $(u,v)$ to some
non-cutting edge since $P^*$ connects $r$ to a leaf of $T$.  Let $U_0$
be the umbrella obtained by adding all other neighbours of $u$ and $v$
to $P$.  For any hanging subgraph $S_{a,b}$ of $P$, the dual tree $T_S$ is a subtree
of $T-P^*$, and therefore $rpw(T_S)\leq H-1$.  By induction, 
$ud(S_{a,b};a,b)\leq H-1$,
and combining the umbrella systems of the hanging subgraphs with $U_0$ hence
shows $ud(G;u,v)\leq H$ as desired.

For the last inequality, we already know that $rpw(T,r)\leq 2pw(T)+1$ for
all choices of the root $r$ \cite{Bie-OPTI-ArXiV}.  To prove the slightly
tighter bound, assume that $P$ is a main path of $T$ and let $r$ be its end.
Root $T$ at $r$, and use $P$ as rpw-main path.
If $T=P$, then $rpw(T)=1$ and we are done.  Else $pw(T)\geq 1$ and
any subtree $T'$ of $T-P$ has $pw(T')\leq pw(T)-1$.  The
bound in \cite{Bie-OPTI-ArXiV} gives $rpw(T') \leq 2pw(T)-1$.  Thus
$rpw(T,r)\leq 1+\max_{T'} rpw(T') \leq 2pw(T)$.
\end{proof}

It is known that $rpw(T)\leq \log(n+1)$ \cite{Bie-OPTI-ArXiV}, and
therefore $bd^{\mathit{free}}(G)\leq \log(n+1)$ as claimed earlier.
All bounds in Lemma~\ref{lem:RootedPathwidthComparison}
are tight, except for a `+1' term. 
Both graphs to show this are constructed as follows.
Define a small graph $G_1$ 
that has a marked root-edge $(u,v)$ and some marked anchor-edges.
Obtain graph $G_i$ by starting with graph $G_1$ and attaching copies
of $G_{i-1}$ such that the root-edge of each $G_{i-1}$ is one of the
anchor-edges of $G_1$.  Let $T_i$ be the dual tree of graph $G_i$.
We leave to the reader to verify the following claims:
\begin{itemize}
\item For the construction using the graph in Figure~\ref{fig:tight}(a),
we have $pw(T_i)\geq 2i$ while $ud(G_i;u,v)\leq i$. Therefore
$$i\leq \frac{1}{2} pw(T_i) \leq bd^{\mathit{free}}(G_i) \leq ud^{\mathit{free}}(G_i)
\leq ud(G_i;u,v)\leq i,$$
and the left two inequalities are tight.
\item For the construction using the graph in Figure~\ref{fig:tight}(b),
we have $pw(T_i)\leq i$ and $ud(G_i;u,v)\geq 2i$.  Therefore 
$$2i\leq ud(G_i;u,v)\leq ud^{\mathit{free}}(G_i)+1\leq 
rpw^{\mathit{free}}(G_i)+1 \leq 2pw(T_i)+1\leq 2i+1,$$
and the third and fourth inequality are tight up to a `$+1$' term.
\end{itemize}

\begin{figure}[ht!]
\hspace*{\fill}
\begin{subfigure}[b]{0.5\linewidth}
\includegraphics[width=80mm]{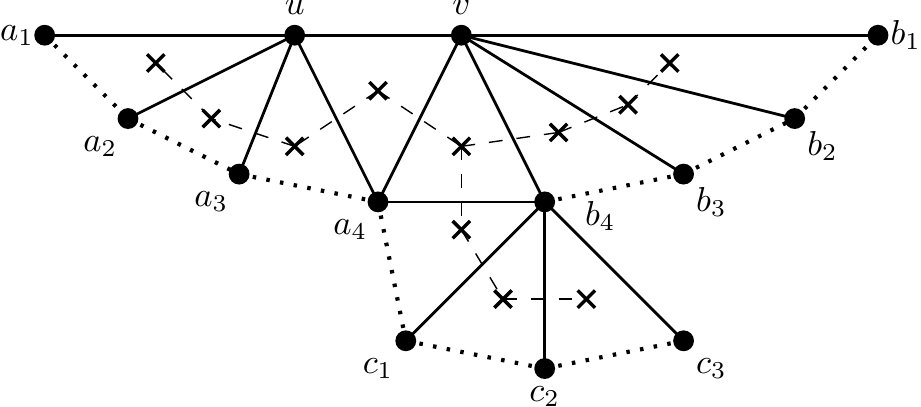}
\caption{}
\end{subfigure}
\hspace*{\fill}
\begin{subfigure}[b]{0.4\linewidth}
\includegraphics[width=60mm]{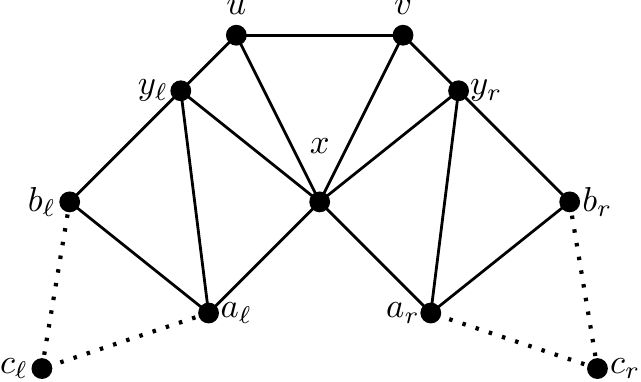}
\caption{}
\end{subfigure}
\hspace*{\fill}
\caption{Constructions to prove tightness.}
\label{fig:RootedPathwidthUpperBound}
\label{fig:PathwidthUpperBound}
\label{fig:PathwidthLowerBound}
\label{fig:tight}
\end{figure}

\end{appendix}
\end{document}